\documentclass[final]{siamltex}
\usepackage{amsmath}
\usepackage{graphicx}

\def\T{{\hbox{\scriptsize{\rm T}}}}
\def\tinyT{{\hbox{\tiny{\rm T}}}}
\def\epsilon{\varepsilon}
\def\phi{\varphi}
\def\bigoh{\mathcal{O}}
\def\th{{\rm th}}
\def\ith{{\it th}}

\def\st{{\rm st}}
\def\ist{{\it st}}
\def\Id{{\bf 1}}
\def\0s{{\bf 0}}

\def\registered{$^{\hbox{\ooalign{\hfil\raise .20ex\hbox{\textbf{\tiny R}}\hfil\crcr\mathhexbox20C}}}$}

\newtheorem{observe}[theorem]{Observation}
\newtheorem{remark1}[theorem]{Remark}

\newenvironment{observation}{\begin{observe} \rm}{\end{observe}}
\newenvironment{remark}{\begin{remark1} \rm}{\end{remark1}}

\title{A randomized algorithm for\\principal component analysis}
\author{Vladimir Rokhlin\thanks{Departments of Computer Science, Mathematics,
and Physics, Yale University, New Haven, CT 06511;
supported in part by DARPA/AFOSR Grant FA9550-07-1-0541.} \and
Arthur Szlam\thanks{Department of Mathematics, UCLA, Los Angeles, CA 90095-1555;
supported in part by NSF Grant DMS-0811203 ({\tt aszlam@math.ucla.edu}).} \and
Mark Tygert\thanks{Department of Mathematics, UCLA, Los Angeles, CA 90095-1555
({\tt tygert@aya.yale.edu}).}
}

\begin{document}

\maketitle

\begin{abstract}
Principal component analysis (PCA) requires the computation
of a low-rank approximation to a matrix containing the data being analyzed.
In many applications of PCA, the best possible accuracy
of any rank-deficient approximation is at most a few digits
(measured in the spectral norm,
relative to the spectral norm of the matrix being approximated).
In such circumstances, efficient algorithms have not come
with guarantees of good accuracy,
unless one or both dimensions of the matrix being approximated are small.
We describe an efficient algorithm for the low-rank approximation of matrices
that produces accuracy very close to the best possible,
for matrices of arbitrary sizes.
We illustrate our theoretical results via several numerical examples.
\end{abstract}

\begin{keywords}
PCA, singular value decomposition, SVD, low rank, Lanczos, power
\end{keywords}

\begin{AMS}
65F15, 65C60, 68W20
\end{AMS}

\pagestyle{myheadings}
\thispagestyle{plain}
\markboth{ROKHLIN, SZLAM, AND TYGERT}{A RANDOMIZED ALGORITHM FOR PCA}

\section{Introduction}

Principal component analysis\,(PCA)\,is among the most\,widely used techniques
in statistics, data analysis, and data mining.
PCA is the basis of many machine learning methods,
including the latent semantic analysis
of large databases of text and HTML documents described
in~\cite{deerwester-dumais-furnas-landauer-harshman}.
Computationally, PCA amounts to the low-rank approximation of a matrix
containing the data being analyzed.
The present article describes an algorithm
for the low-rank approximation of matrices, suitable for PCA.
This paper demonstrates both theoretically and via numerical examples
that the algorithm efficiently produces low-rank approximations
whose accuracies are very close to the best possible.

The canonical construction of the best possible rank-$k$ approximation
to a real $m \times n$ matrix $A$ uses the singular value decomposition (SVD)
of $A$,
\begin{equation}
\label{full_svd}
A = U \, \Sigma \, V^\T,
\end{equation}
where $U$ is a real unitary $m \times m$ matrix,
$V$ is a real unitary $n \times n$ matrix,
and $\Sigma$ is a real $m \times n$ matrix whose only nonzero entries
appear in nonincreasing order on the diagonal and are nonnegative.
The diagonal entries $\sigma_1$,~$\sigma_2$,
\dots, $\sigma_{\min(m,n)-1}$,~$\sigma_{\min(m,n)}$
of $\Sigma$ are known as the singular values of $A$.
The best rank-$k$ approximation to $A$, with $k < m$ and $k < n$, is
\begin{equation}
\label{low_rank_approx}
A \approx \tilde{U} \, \tilde{\Sigma} \, \tilde{V}^\T,
\end{equation}
where $\tilde{U}$ is the leftmost $m \times k$ block of $U$,
$\tilde{V}$ is the leftmost $n \times k$ block of $V$,
and $\tilde{\Sigma}$ is the $k \times k$ matrix
whose only nonzero entries appear in nonincreasing order on the diagonal
and are the $k$ greatest singular values of $A$.
This approximation is ``best'' in the sense that
the spectral norm $\| A - B \|$ of the difference between $A$
and a rank-$k$ matrix $B$ is minimal
for $B = \tilde{U} \, \tilde{\Sigma} \, \tilde{V}^\T$.
In fact,
\begin{equation}
\| A - \tilde{U} \, \tilde{\Sigma} \, \tilde{V}^\T \| = \sigma_{k+1},
\end{equation}
where $\sigma_{k+1}$ is the $(k+1)^\st$ greatest singular value of $A$.
For more information about the SVD, see, for example,
Chapter~8 in~\cite{golub-van_loan}.

For definiteness, let us assume that $m \le n$
and that $A$ is an arbitrary (dense) real $m \times n$ matrix.
To compute a rank-$k$ approximation to $A$,
one might form the matrices $U$, $\Sigma$, and $V$ in~(\ref{full_svd}),
and then use them to construct $\tilde{U}$, $\tilde{\Sigma}$, and $\tilde{V}$
in~(\ref{low_rank_approx}).
However, even computing just $\Sigma$, the leftmost $m$ columns of $U$,
and the leftmost $m$ columns of $V$ requires at least
$\bigoh(n m^2)$ floating-point operations (flops) using any
of the standard algorithms
(see, for example, Chapter~8 in~\cite{golub-van_loan}).
Alternatively, one might use pivoted $QR$-decomposition algorithms,
which require $\bigoh(nmk)$ flops
and typically produce a rank-$k$ approximation $B$ to $A$ such that
\begin{equation}
\label{gu_bound}
\| A - B \| \le 10 \sqrt{m} \; \sigma_{k+1},
\end{equation}
where $\|A-B\|$ is the spectral norm of $A-B$,
and $\sigma_{k+1}$ is the $(k+1)^\st$ greatest singular value of $A$
(see, for example, Chapter~5 in~\cite{golub-van_loan}).
Furthermore, the algorithms of~\cite{gu-eisenstat} require only
about $\bigoh(nmk)$ flops to produce a rank-$k$ approximation that
(unlike an approximation produced by a pivoted $QR$-decomposition)
has been guaranteed to satisfy a bound nearly as strong as~(\ref{gu_bound}).

While the accuracy in~(\ref{gu_bound}) is sufficient
for many applications of low-rank approximation,
PCA often involves $m \ge$ 10,000,
and a ``signal-to-noise ratio'' $\sigma_1/\sigma_{k+1} \le 100$,
where $\sigma_1 = \|A\|$ is the greatest singular value of $A$,
and $\sigma_{k+1}$ is the $(k+1)^\st$ greatest.
Moreover, the singular values $\le \sigma_{k+1}$
often arise from noise in the process generating the data in $A$,
making the singular values of $A$ decay so slowly that
$\sigma_m \ge \sigma_{k+1}/10$.
When $m \ge$ 10,000, $\sigma_1/\sigma_{k+1} \le 100$,
and $\sigma_m \ge \sigma_{k+1}/10$, the rank-$k$ approximation $B$ produced
by a pivoted $QR$-decomposition algorithm
typically satisfies $\| A - B \| \sim \| A \|$
--- the ``approximation'' $B$ is effectively unrelated
to the matrix $A$ being approximated!
For large matrices whose ``signal-to-noise ratio''
$\sigma_1/\sigma_{k+1}$ is less than 10,000,
the $\sqrt{m}$ factor in~(\ref{gu_bound}) may be unacceptable.
Now, pivoted $QR$-decomposition algorithms are not the only algorithms
which can compute a rank-$k$ approximation using $\bigoh(nmk)$ flops.
However, other algorithms, such as those of
\cite{achlioptas-mcsherry0}, \cite{achlioptas-mcsherry}, \cite{chan-hansen},
\cite{clarkson-woodruff}, \cite{deshpande-rademacher-vempala-wang},
\cite{deshpande-vempala}, \cite{drineas-drinea-huggins},
\cite{drineas-kannan-mahoney2}, \cite{drineas-kannan-mahoney3},
\cite{drineas-mahoney-muthukrishnan1}, \cite{drineas-mahoney-muthukrishnan2},
\cite{friedland-kaveh-niknejad-zare}, \cite{frieze-kannan},
\cite{frieze-kannan-vempala0}, \cite{frieze-kannan-vempala},
\cite{goreinov-tyrtyshnikov}, \cite{goreinov-tyrtyshnikov-zamarashkin2},
\cite{goreinov-tyrtyshnikov-zamarashkin1}, \cite{gu-eisenstat},
\cite{har-peled},
\cite{liberty-woolfe-martinsson-rokhlin-tygert}, \cite{mahoney-drineas},
\cite{papadimitriou-raghavan-tamaki-vempala},
\cite{sarlos3}, \cite{sarlos4}, \cite{sun-xie-zhang-faloutsos},
\cite{tyrtyshnikov}, and~\cite{woolfe-liberty-rokhlin-tygert},
also yield accuracies involving factors of at least $\sqrt{m}$
when the singular values $\sigma_{k+1}$, $\sigma_{k+2}$, $\sigma_{k+3}$, \dots\
of $A$ decay slowly.
(The decay is rather slow if, for example,
$\sigma_{k+j} \sim j^\alpha \, \sigma_{k+1}$
for $j = 1$,~$2$,~$3$, \dots, with $-1/2 < \alpha \le 0$.
Many of these other algorithms are designed to produce approximations
having special properties not treated in the present paper,
and their spectral-norm accuracy is good when the singular values decay
sufficiently fast. Fairly recent surveys of algorithms
for low-rank approximation are available in~\cite{sarlos3}, \cite{sarlos4},
and~\cite{liberty-woolfe-martinsson-rokhlin-tygert}.)

The algorithm described in the present paper produces
a rank-$k$ approximation $B$ to $A$ such that
\begin{equation}
\label{very_rough}
\| A - B \| \le C \, m^{1/(4i+2)} \, \sigma_{k+1}
\end{equation}
with very high probability (typically $1-10^{-15}$, independent of $A$,
with the choice of parameters from Remark~\ref{par_remark} below),
where $\|A-B\|$ is the spectral norm of $A-B$,
$i$ is a nonnegative integer specified by the user,
$\sigma_{k+1}$ is the $(k+1)^\st$ greatest singular value of $A$,
and $C$ is a constant independent of $A$
that theoretically may depend on the parameters of the algorithm.
(Numerical evidence such as that in Section~\ref{numerical}
suggests at the very least that $C < 10$;
(\ref{explicit_eval}) and~(\ref{the_point})
in Section~\ref{algorithm} provide more complicated theoretical bounds on $C$.)
The algorithm requires $\bigoh(nmki)$ floating-point operations when $i>0$.
In many applications of PCA, $i = 1$ or $i = 2$ is sufficient,
and the algorithm then requires only $\bigoh(nmk)$ flops.
The algorithm provides the rank-$k$ approximation $B$ in the form of an SVD,
outputting three matrices, $\tilde{U}$, $\tilde{\Sigma}$, and $\tilde{V}$,
such that $B = \tilde{U} \, \tilde{\Sigma} \, \tilde{V}^\T$,
where the columns of $\tilde{U}$ are orthonormal,
the columns of $\tilde{V}$ are orthonormal,
and the entries of $\tilde{\Sigma}$ are all nonnegative
and zero off the diagonal.

The algorithm of the present paper is randomized,
but succeeds with very high probability;
for example, the bound~(\ref{explicit_eval}) on its accuracy holds
with probability greater than $1-10^{-15}$.
The algorithm is similar to many recently discussed randomized algorithms
for low-rank approximation, but produces approximations of higher accuracy
when the singular values $\sigma_{k+1}$, $\sigma_{k+2}$, $\sigma_{k+3}$, \dots\
of the matrix being approximated decay slowly; see, for example, \cite{sarlos3}
or~\cite{liberty-woolfe-martinsson-rokhlin-tygert}.
The algorithm is a variant of that in~\cite{roweis},
and the analysis of the present paper should extend to the algorithm
of~\cite{roweis}; \cite{roweis} stimulated the authors' collaboration.
The algorithm may be regarded as a generalization
of the randomized power methods of~\cite{dixon}
and~\cite{kuczynski-wozniakowski},
and in fact we use the latter to ascertain the approximations' accuracy
rapidly and reliably.

The algorithm admits obvious ``out-of-core'' and parallel implementations
(assuming that the user chooses the parameter $i$ in~(\ref{very_rough})
to be reasonably small).
As with the algorithms of~\cite{dixon}, \cite{kuczynski-wozniakowski},
\cite{liberty-woolfe-martinsson-rokhlin-tygert},
\cite{martinsson-rokhlin-tygert3}, \cite{roweis},
\cite{sarlos3}, and~\cite{sarlos4},
the core steps of the algorithm of the present paper
involve the application of the matrix $A$ being approximated
and its transpose $A^\T$ to random vectors.
The algorithm is more efficient when $A$ and $A^\T$ can be applied rapidly
to arbitrary vectors, such as when $A$ is sparse.

Throughout the present paper, we use $\Id$ to denote an identity matrix.
We use $\0s$ to denote a matrix whose entries are all zeros.
For any matrix $A$, we use $\|A\|$ to denote the spectral norm of $A$,
that is, $\|A\|$ is the greatest singular value of $A$.
Furthermore, the entries of all matrices in the present paper are real valued,
though the algorithm and analysis extend trivially to matrices
whose entries are complex valued.

The present paper has the following structure:
Section~\ref{prelims} collects together various known facts
which later sections utilize.
Section~\ref{apparatus} provides the principal lemmas used in bounding
the accuracy of the algorithm in Section~\ref{algorithm}.
Section~\ref{algorithm} describes the algorithm of the present paper.
Section~\ref{numerical} illustrates the performance of the algorithm
via several numerical examples.
The appendix, Section~\ref{appendix}, proves two lemmas stated earlier
in Section~\ref{apparatus}.
We encourage the reader to begin with Sections~\ref{algorithm}
and~\ref{numerical}, referring back to the relevant portions
of Sections~\ref{prelims} and~\ref{apparatus} as they are referenced.

\section{Preliminaries}
\label{prelims}

In this section, we summarize various facts about matrices and functions.
Subsection~\ref{general_singular_values} discusses the singular values
of arbitrary matrices. Subsection~\ref{random_singular_values}
discusses the singular values of certain random matrices.
Subsection~\ref{monotone} observes that a certain function is monotone.

\subsection{Singular values of general matrices}
\label{general_singular_values}

The following trivial technical lemma will be needed
in Section~\ref{apparatus}.

\begin{lemma}
Suppose that $m$ and $n$ are positive integers with $m \ge n$.
Suppose further that $A$ is a real $m \times n$ matrix
such that the least (that is, the $n^\ith$ greatest) singular value $\sigma_n$
of $A$ is nonzero.

Then,
\begin{equation}
\label{pseudoinverse_norm}
\left\| (A^\T \, A)^{-1} \, A^\T \right\| = \frac{1}{\sigma_n}.
\end{equation}
\end{lemma}

The following lemma states that the greatest singular value of a matrix $A$
is at least as large as the greatest singular value
of any rectangular block of entries in $A$;
the lemma is a straightforward consequence
of the minimax properties of singular values
(see, for example, Section~47 of Chapter~2 in~\cite{wilkinson}).

\begin{lemma}
\label{minimax_consequence}
Suppose that $k$, $l$, $m$, and~$n$ are positive integers
with $k \le m$ and $l \le n$.
Suppose further that $A$ is a real $m \times n$ matrix,
and $B$ is a $k \times l$ rectangular block of entries in $A$.

Then, the greatest singular value of $B$ is at most
the greatest singular value of $A$.
\end{lemma}

The following classical lemma provides an approximation $Q \, S$
to an $n \times l$ matrix $R$
via an $n \times k$ matrix $Q$ whose columns are orthonormal,
and a $k \times l$ matrix $S$.
As remarked in Observation~\ref{least_squares},
the proof of this lemma provides a classic algorithm for computing $Q$ and $S$,
given $R$. We include the proof since we will be using this algorithm.

\begin{lemma}
Suppose that $k$, $l$, and $n$ are positive integers with $k < l \le n$,
and $R$ is a real $n \times l$ matrix.

Then, there exist a real $n \times k$ matrix $Q$
whose columns are orthonormal,
and a real $k \times l$ matrix $S$, such that
\begin{equation}
\label{svd_qr}
\| Q \, S - R \| \le \rho_{k+1},
\end{equation}
where $\rho_{k+1}$ is the $(k+1)^\ist$ greatest singular value of $R$.
\end{lemma}

\begin{proof}
We start by forming an SVD of $R$,
\begin{equation}
\label{little_svd}
R = U \, \Sigma \, V^\T,
\end{equation}
where $U$ is a real $n \times l$ matrix whose columns are orthonormal,
$V$ is a real $l \times l$ matrix whose columns are orthonormal,
and $\Sigma$ is a real diagonal $l \times l$ matrix, such that
\begin{equation}
\label{little_ordering}
\Sigma_{j,j} = \rho_j
\end{equation}
for $j = 1$,~$2$, \dots, $l-1$,~$l$,
where $\Sigma_{j,j}$ is the entry in row $j$ and column $j$ of $\Sigma$,
and $\rho_j$ is the $j^\th$ greatest singular value of $R$.
We define $Q$ to be the leftmost $n \times k$ block of $U$,
and $P$ to be the rightmost $n \times (l-k)$ block of $U$, so that
\begin{equation}
\label{left_sing}
U = \left( \begin{array}{c|c} Q & P \end{array} \right).
\end{equation}
We define $S$ to be the uppermost $k \times l$ block of $\Sigma \, V^\T$,
and $T$ to be the lowermost $(l-k) \times l$ block of $\Sigma \, V^\T$,
so that
\begin{equation}
\label{right_sing}
\Sigma \, V^\T = \left( \begin{array}{c} S \\\hline T \end{array} \right).
\end{equation}
Combining~(\ref{little_svd}), (\ref{little_ordering}),
(\ref{left_sing}), (\ref{right_sing}),
and the fact that the columns of $U$ are orthonormal,
as are the columns of $V$, yields~(\ref{svd_qr}).
\end{proof}

\begin{observation}
\label{least_squares}
In order to compute the matrices $Q$ and $S$ in~(\ref{svd_qr})
from the matrix $R$,
we can construct~(\ref{little_svd}),
and then form $Q$ and $S$
according to~(\ref{left_sing}) and~(\ref{right_sing}).
(See, for example, Chapter~8 in~\cite{golub-van_loan} for details
concerning the computation of the SVD.)
\end{observation}

\subsection{Singular values of random matrices}
\label{random_singular_values}

The following lemma provides a highly probable upper bound
on the greatest singular value
of a square matrix whose entries are independent, identically distributed
(i.i.d.) Gaussian random variables of zero mean and unit variance;
Formula~8.8 in~\cite{goldstine-von_neumann} provides an equivalent formulation
of the lemma.

\begin{lemma}
\label{greatest_bound}
Suppose that $n$ is a positive integer,
$G$ is a real $n \times n$ matrix whose entries are
i.i.d.\ Gaussian random variables of zero mean and unit variance,
and $\gamma$ is a positive real number, such that $\gamma > 1$ and
\begin{equation}
\label{failure_prob}
1 - \frac{1}{4 \, (\gamma^2-1) \, \sqrt{\pi n \gamma^2}}
    \left( \frac{2 \gamma^2}{e^{\gamma^2-1}} \right)^n
\end{equation}
is nonnegative.

Then, the greatest singular value of $G$ is at most $\sqrt{2n} \, \gamma$
with probability not less than the amount in~(\ref{failure_prob}).
\end{lemma}

Combining Lemmas~\ref{minimax_consequence} and~\ref{greatest_bound}
yields the following lemma,
providing a highly probable upper bound on the greatest singular value
of a rectangular matrix whose entries are i.i.d.\ Gaussian
random variables of zero mean and unit variance.

\begin{lemma}
\label{greatest_value}
Suppose that $l$, $m$, and $n$ are positive integers
with $n \ge l$ and $n \ge m$.
Suppose further that $G$ is a real $l \times m$ matrix whose entries are
i.i.d.\ Gaussian random variables of zero mean and unit variance,
and $\gamma$ is a positive real number, such that
$\gamma > 1$ and~(\ref{failure_prob}) is nonnegative.

Then, the greatest singular value of $G$ is at most $\sqrt{2n} \, \gamma$
with probability not less than the amount in~(\ref{failure_prob}).
\end{lemma}

The following lemma provides a highly probable lower bound
on the least singular value
of a rectangular matrix whose entries are i.i.d.\ Gaussian
random variables of zero mean and unit variance;
Formula~2.5 in~\cite{chen-dongarra}
and the proof of Lemma~4.1 in~\cite{chen-dongarra}
together provide an equivalent formulation of Lemma~\ref{least_value}.

\begin{lemma}
\label{least_value}
Suppose that $j$ and $l$ are positive integers with $j \le l$.
Suppose further that $G$ is a real $l \times j$ matrix whose entries are
i.i.d.\ Gaussian random variables of zero mean and unit variance,
and $\beta$ is a positive real number, such that
\begin{equation}
\label{failure_prob2}
1 - \frac{1}{\sqrt{2 \pi \, (l-j+1)}}
 \, \left( \frac{e}{(l-j+1) \, \beta} \right)^{l-j+1}
\end{equation}
is nonnegative.

Then, the least (that is, the $j^\ith$ greatest) singular value
of $G$ is at least $1 / (\sqrt{l} \; \beta)$
with probability not less than the amount in~(\ref{failure_prob2}).
\end{lemma}

\subsection{A monotone function}
\label{monotone}

The following technical lemma will be needed
in Section~\ref{algorithm}.

\begin{lemma}
\label{monotonicity}
Suppose that $\alpha$ is a nonnegative real number,
and $f$ is the function defined on $(0,\infty)$ via the formula
\begin{equation}
f(x) = \frac{1}{\sqrt{2 \pi x}} \left( \frac{e\alpha}{x} \right)^x.
\end{equation}

Then, $f$ decreases monotonically for $x > \alpha$.
\end{lemma}

\begin{proof}
The derivative of $f$ is
\begin{equation}
\label{derivative}
f'(x) = f(x) \left( \ln\left(\frac{\alpha}{x}\right) - \frac{1}{2x} \right)
\end{equation}
for any positive real number $x$.
The right-hand side of~(\ref{derivative}) is negative when $x > \alpha$.
\end{proof}

\section{Mathematical apparatus}
\label{apparatus}

In this section, we provide lemmas to be used in Section~\ref{algorithm}
in bounding the accuracy of the algorithm of the present paper.

The following lemma, proven in the appendix (Section~\ref{appendix}),
shows that the product $A \, Q \, Q^\T$
of matrices $A$, $Q$, and $Q^\T$
is a good approximation to a matrix $A$,
provided that there exist matrices $G$ and $S$ such that
\begin{enumerate}
\item[1.] the columns of $Q$ are orthonormal,
\item[2.] $Q \, S$ is a good approximation to $(G \, (A \, A^\T)^i \, A)^\T$,
and
\item[3.] there exists a matrix $F$ such that $\| F \|$ is not too large,
and $F \, G \, (A \, A^\T)^i \, A$ is a good approximation to $A$.
\end{enumerate}

\begin{lemma}
\label{all_together2}
Suppose that $i$, $k$, $l$, $m$, and~$n$ are positive integers
with $k \le l \le m \le n$.
Suppose further that $A$ is a real $m \times n$ matrix,
$Q$ is a real $n \times k$ matrix whose columns are orthonormal,
$S$ is a real $k \times l$ matrix,
$F$ is a real $m \times l$ matrix,
and $G$ is a real $l \times m$ matrix.

Then,
\begin{equation}
\label{reconstruction2}
\| A \, Q \, Q^\T - A \|
\le 2 \, \| F \, G \, (A \, A^\T)^i \, A - A \|
  + 2 \, \| F \| \, \| Q \, S - (G \, (A \, A^\T)^i \, A)^\T \|.
\end{equation}
\end{lemma}

The following lemma, proven in the appendix (Section~\ref{appendix}),
states that,
for any positive integer $i$, matrix $A$, and matrix $G$ whose entries are
i.i.d.\ Gaussian random variables of zero mean and unit variance,
with very high probability there exists a matrix $F$
with a reasonably small norm,
such that $F \, G \, (A \, A^\T)^i \, A$ is a good approximation to $A$.
This lemma is similar to Lemma~19 of~\cite{martinsson-rokhlin-tygert3}.

\begin{lemma}
\label{probability_bounds2}
Suppose that $i$, $j$, $k$, $l$, $m$, and~$n$ are positive integers
with $j < k < l < m \le n$.
Suppose further that $A$ is a real $m \times n$ matrix,
$G$ is a real $l \times m$ matrix whose entries are
i.i.d.\ Gaussian random variables of zero mean and unit variance,
and $\beta$ and $\gamma$ are positive real numbers, such that
the $j^\ith$ greatest singular value $\sigma_j$ of $A$ is positive,
$\gamma > 1$, and
\begin{multline}
\label{probability2}
\Phi
  = 1 - \frac{1}{\sqrt{2 \pi \, (l-j+1)}}
 \, \left( \frac{e}{(l-j+1) \, \beta} \right)^{l-j+1} \\
  - \frac{1}{4 \, (\gamma^2-1) \, \sqrt{\pi \, \max(m-k,l) \; \gamma^2}}
    \left( \frac{2 \gamma^2}{e^{\gamma^2-1}} \right)^{\max(m-k,\,l)} \\
  - \frac{1}{4 \, (\gamma^2-1) \, \sqrt{\pi \, l \, \gamma^2}}
    \left( \frac{2 \gamma^2}{e^{\gamma^2-1}} \right)^l
\end{multline}
is nonnegative.

Then, there exists a real $m \times l$ matrix $F$ such that
\begin{multline}
\label{approximation2}
\| F \, G \, (A \, A^\T)^i \, A - A \|
\le \sqrt{ 2 l^2 \, \beta^2 \, \gamma^2 + 1 }
 \;\; \sigma_{j+1} \\
  + \sqrt{ 2 l \, \max(m-k,l) \, \beta^2 \, \gamma^2
        \, \left( \frac{\sigma_{k+1}}{\sigma_j} \right)^{4i} + 1 }
 \;\; \sigma_{k+1}
\end{multline}
and
\begin{equation}
\label{small_norm2}
\| F \| \le \frac{\sqrt{l} \; \beta}{(\sigma_j)^{2i}}
\end{equation}
with probability not less than $\Phi$ defined in~(\ref{probability2}),
where $\sigma_j$ is the $j^\ith$ greatest singular value of $A$,
$\sigma_{j+1}$ is the $(j+1)^\ist$ greatest singular value of $A$,
and $\sigma_{k+1}$ is the $(k+1)^\ist$ greatest singular value of $A$.
\end{lemma}

Given a matrix $A$,
and a matrix $G$ whose entries are i.i.d.\ Gaussian random variables
of zero mean and unit variance,
the following lemma provides a highly probable upper bound
on the singular values of the product $G \, A$
in terms of the singular values of $A$.
This lemma is reproduced from~\cite{martinsson-rokhlin-tygert3},
where it appears as Lemma~20.

\begin{lemma}
\label{singular_value_stretching}
Suppose that $j$, $k$, $l$, $m$, and~$n$ are positive integers
with $k < l$, such that $k + j < m$ and $k + j < n$.
Suppose further that $A$ is a real $m \times n$ matrix,
$G$ is a real $l \times m$ matrix whose entries are
i.i.d.\ Gaussian random variables of zero mean and unit variance,
and $\gamma$ is a positive real number, such that
$\gamma > 1$ and
\begin{multline}
\label{probability3}
\Xi
  = 1 - \frac{1}{4 \, (\gamma^2-1) \, \sqrt{\pi \, \max(m-k-j,l) \, \gamma^2}}
    \left( \frac{2 \gamma^2}{e^{\gamma^2-1}} \right)^{\max(m-k-j,\,l)} \\
  - \frac{1}{4 \, (\gamma^2-1) \, \sqrt{\pi \, \max(k+j,l) \; \gamma^2}}
    \left( \frac{2 \gamma^2}{e^{\gamma^2-1}} \right)^{\max(k+j,\,l)}
\end{multline}
is nonnegative.

Then,
\begin{equation}
\label{stretched_singular_value}
\rho_{k+1} \le \sqrt{2 \, \max(k+j,l)} \; \gamma \; \sigma_{k+1}
             + \sqrt{2 \, \max(m-k-j,l)} \; \gamma \; \sigma_{k+j+1}
\end{equation}
with probability not less than $\Xi$ defined in~(\ref{probability3}),
where $\rho_{k+1}$ is the $(k+1)^\ist$ greatest singular value of $G \, A$,
$\sigma_{k+1}$ is the $(k+1)^\ist$ greatest singular value of $A$,
and $\sigma_{k+j+1}$ is the $(k+j+1)^\ist$ greatest singular value of $A$.
\end{lemma}

The following corollary follows immediately from the preceding lemma,
by replacing the matrix $A$ with $(A \, A^\T)^i \, A$,
the integer $k$ with $j$, and the integer $j$ with $k-j$.

\begin{corollary}
\label{singular_value_stretching2}
Suppose $i$, $j$, $k$, $l$, $m$, and~$n$ are positive integers
with $j < k < l < m \le n$.
Suppose further that $A$ is a real $m \times n$ matrix,
$G$ is a real $l \times m$ matrix whose entries are
i.i.d.\ Gaussian random variables of zero mean and unit variance,
and $\gamma$ is a positive real number, such that
$\gamma > 1$ and
\begin{multline}
\label{probability32}
\Psi
  = 1 - \frac{1}{4 \, (\gamma^2-1) \, \sqrt{\pi \, \max(m-k,l) \, \gamma^2}}
    \left( \frac{2 \gamma^2}{e^{\gamma^2-1}} \right)^{\max(m-k,\,l)} \\
  - \frac{1}{4 \, (\gamma^2-1) \, \sqrt{\pi \, l \; \gamma^2}}
    \left( \frac{2 \gamma^2}{e^{\gamma^2-1}} \right)^l
\end{multline}
is nonnegative.

Then,
\begin{equation}
\label{stretched_singular_value2}
\rho_{j+1} \le \sqrt{2 l} \; \gamma \; (\sigma_{j+1})^{2i+1}
             + \sqrt{2 \, \max(m-k,l)} \; \gamma \; (\sigma_{k+1})^{2i+1}
\end{equation}
with probability not less than $\Psi$ defined in~(\ref{probability32}),
where $\rho_{j+1}$ is the $(j+1)^\ist$ greatest singular value
of $G \, (A \, A^\T)^i \, A$,
$\sigma_{j+1}$ is the $(j+1)^\ist$ greatest singular value of $A$,
and $\sigma_{k+1}$ is the $(k+1)^\ist$ greatest singular value of $A$.
\end{corollary}

\section{Description of the algorithm}
\label{algorithm}

In this section, we describe the algorithm of the present paper,
providing details about its accuracy and computational costs.
Subsection~\ref{main_algorithm} describes the basic algorithm.
Subsection~\ref{costs} tabulates the computational costs of the algorithm.
Subsection~\ref{modified} describes a complementary algorithm.
Subsection~\ref{blanczos} describes a computationally more expensive variant
that is somewhat more accurate and tolerant to roundoff.

\subsection{The algorithm}
\label{main_algorithm}

Suppose that $i$, $k$, $m$, and $n$ are positive integers
with $2k < m \le n$, and $A$ is a real $m \times n$ matrix.
In this subsection, we will construct an approximation to an SVD of $A$
such that
\begin{equation}
\label{sort_of_svd}
\| A - U \, \Sigma \, V^\T \| \le C \, m^{1/(4i+2)} \, \sigma_{k+1}
\end{equation}
with very high probability,
where $U$ is a real $m \times k$ matrix
whose columns are orthonormal,
$V$ is a real $n \times k$ matrix whose columns are orthonormal,
$\Sigma$ is a real diagonal $k \times k$ matrix
whose entries are all nonnegative,
$\sigma_{k+1}$ is the $(k+1)^\st$ greatest singular value of $A$,
and $C$ is a constant independent of $A$ that depends on the parameters
of the algorithm.
(Section~\ref{numerical} will give an empirical indication of the size of $C$,
and~(\ref{explicit_eval}) will give one of our best theoretical estimates
to date.)

Intuitively, we could apply $A^\T$ to several random vectors,
in order to identify the part of its range corresponding
to the larger singular values.
To enhance the decay of the singular values,
we apply $A^\T \, (A \, A^\T)^i$ instead.
Once we have identified most of the range of $A^\T$,
we perform several linear-algebraic manipulations in order to recover
an approximation to $A$.
(It is possible to obtain a similar, somewhat less accurate algorithm
by substituting our short, fat matrix $A$ for $A^\T$, and $A^\T$ for $A$.)

More precisely, we choose an integer $l > k$ such that $l \le m-k$
(for example, $l = k + 12$), and make the following five steps:

\begin{enumerate}
\item[1.] Using a random number generator,
form a real $l \times m$ matrix $G$ whose entries are
i.i.d.\ Gaussian random variables of zero mean and unit variance,
and compute the $l \times n$ product matrix
\begin{equation}
\label{product2}
R = G \, (A \, A^\T)^i \, A.
\end{equation}
\item[2.] Using an SVD,
form a real $n \times k$ matrix $Q$ whose columns are orthonormal,
such that there exists a real $k \times l$ matrix $S$ for which
\begin{equation}
\label{good_approx2}
\| Q \, S - R^\T \| \le \rho_{k+1},
\end{equation}
where $\rho_{k+1}$ is the $(k+1)^\st$ greatest singular value of $R$.
(See Observation~\ref{least_squares} for details concerning
the construction of such a matrix $Q$.)
\item[3.] Compute the $m \times k$ product matrix
\begin{equation}
\label{product_t}
T = A \, Q.
\end{equation}
\item[4.] Form an SVD of $T$,
\begin{equation}
\label{svd_small}
T = U \, \Sigma \, W^\T,
\end{equation}
where $U$ is a real $m \times k$ matrix whose columns are orthonormal,
$W$ is a real $k \times k$ matrix whose columns are orthonormal,
and $\Sigma$ is a real diagonal $k \times k$ matrix
whose entries are all nonnegative.
(See, for example, Chapter~8 in~\cite{golub-van_loan} for details
concerning the construction of such an SVD.)
\item[5.] Compute the $n \times k$ product matrix
\begin{equation}
\label{product3}
V = Q \, W.
\end{equation}
\end{enumerate}

The following theorem states precisely
that the matrices $U$, $\Sigma$, and $V$ satisfy~(\ref{sort_of_svd}).
See~(\ref{explicit_eval}) for a more compact (but less general) formulation.

\begin{theorem}
\label{the_theorem}
Suppose that $i$, $k$, $l$, $m$, and $n$ are positive integers
with $k < l \le m-k$ and $m \le n$, and $A$ is a real $m \times n$ matrix.
Suppose further that $\beta$ and $\gamma$ are positive real numbers
such that $\gamma>1$,
\begin{equation}
\label{monotonicity_assump}
(l-k+1) \, \beta \ge 1,
\end{equation}
\begin{equation}
\label{simplifying_assump}
2 \, l^2 \, \gamma^2 \, \beta^2 \ge 1,
\end{equation}
and
\begin{multline}
\label{final_prob}
\Pi = 1 - \frac{1}{2 \, (\gamma^2-1) \, \sqrt{\pi \, (m-k) \, \gamma^2}}
      \left( \frac{2 \gamma^2}{e^{\gamma^2-1}} \right)^{m-k}
    - \frac{1}{2 \, (\gamma^2-1) \, \sqrt{\pi \, l \; \gamma^2}}
      \left( \frac{2 \gamma^2}{e^{\gamma^2-1}} \right)^l \\
    - \frac{1}{\sqrt{2 \pi \, (l-k+1)}}
   \, \left( \frac{e}{(l-k+1) \, \beta} \right)^{l-k+1}
\end{multline}
is nonnegative.
Suppose in addition that $U$, $\Sigma$, and $V$ are the matrices
produced via the five-step algorithm of the present subsection, given above.

Then,
\begin{equation}
\label{the_point}
\| A - U \, \Sigma \, V^\T \| \le 16 \, \gamma \, \beta \, l
\, \left(\frac{m-k}{l}\right)^{1/(4i+2)} \, \sigma_{k+1}
\end{equation}
with probability not less than $\Pi$,
where $\Pi$ is defined in~(\ref{final_prob}),
and $\sigma_{k+1}$ is the $(k+1)^\ist$ greatest singular value of $A$.
\end{theorem}

\begin{proof}
Observing that $U \, \Sigma \, V^\T = A \, Q \, Q^\T$,
it is sufficient to prove that
\begin{equation}
\label{intermediate_step}
\| A \, Q \, Q^\T - A \| \le 16 \, \gamma \, \beta \, l
\, \left(\frac{m-k}{l}\right)^{1/(4i+2)} \, \sigma_{k+1}
\end{equation}
with probability $\Pi$,
where $Q$ is the matrix from~(\ref{good_approx2}),
since combining~(\ref{intermediate_step}), (\ref{product_t}),
(\ref{svd_small}), and~(\ref{product3}) yields~(\ref{the_point}).
We now prove~(\ref{intermediate_step}).

First, we consider the case when
\begin{equation}
\label{first_case}
\| A \| \le \left(\frac{m-k}{l}\right)^{1/(4i+2)} \, \sigma_{k+1}.
\end{equation}
Clearly,
\begin{equation}
\label{triangle_submult}
\| A \, Q \, Q^\T - A \| \le \| A \| \, \| Q \| \, \| Q^\T \| + \| A \|.
\end{equation}
But, it follows from the fact that the columns of $Q$ are orthonormal that
\begin{equation}
\label{normortho1}
\| Q \| \le 1
\end{equation}
and
\begin{equation}
\label{normortho2}
\| Q^\T \| \le 1.
\end{equation}
Combining~(\ref{triangle_submult}), (\ref{normortho1}), (\ref{normortho2}),
(\ref{first_case}), and~(\ref{simplifying_assump})
yields~(\ref{intermediate_step}), completing the proof
for the case when~(\ref{first_case}) holds.

For the remainder of the proof, we consider the case when
\begin{equation}
\label{second_case}
\| A \| > \left(\frac{m-k}{l}\right)^{1/(4i+2)} \, \sigma_{k+1}.
\end{equation}
To prove~(\ref{intermediate_step}),
we will use~(\ref{reconstruction2})
(which is restated and proven in Lemma~\ref{all_together22} in the appendix),
namely,
\begin{equation}
\label{basic_bound}
\| A \, Q \, Q^\T - A \|
\le 2 \, \| F \, G \, (A \, A^\T)^i \, A - A \|
  + 2 \, \| F \| \, \| Q \, S - (G \, (A \, A^\T)^i \, A)^\T \|
\end{equation}
for any real $m \times l$ matrix $F$,
where $G$ is from~(\ref{product2}),
and $Q$ and $S$ are from~(\ref{good_approx2}).
We now choose an appropriate matrix $F$.

First, we define $j$ to be the positive integer such that
\begin{equation}
\label{reduced_rank}
\sigma_{j+1} \le \left(\frac{m-k}{l}\right)^{1/(4i+2)} \, \sigma_{k+1}
               < \sigma_j,
\end{equation}
where $\sigma_j$ is the $j^\th$ greatest singular value of $A$,
and $\sigma_{j+1}$ is the $(j+1)^\st$ greatest
(such an integer $j$ exists due to~(\ref{second_case})
and the supposition of the theorem that $l \le m-k$).
We then use the matrix $F$ from~(\ref{approximation2})
and~(\ref{small_norm2}) associated with this integer $j$, so that
(as stated in~(\ref{approximation2}) and~(\ref{small_norm2}),
which are restated and proven in Lemma~\ref{probability_bounds22}
in the appendix)
\begin{multline}
\label{number1}
\| F \, G \, (A \, A^\T)^i \, A - A \|
\le \sqrt{ 2 l^2 \, \beta^2 \, \gamma^2 + 1 }
 \;\; \sigma_{j+1} \\
  + \sqrt{ 2 l \, \max(m-k,l) \, \beta^2 \, \gamma^2
        \, \left( \frac{\sigma_{k+1}}{\sigma_j} \right)^{4i} + 1 }
 \;\; \sigma_{k+1}
\end{multline}
and
\begin{equation}
\label{number2}
\| F \| \le \frac{\sqrt{l} \; \beta}{(\sigma_j)^{2i}}
\end{equation}
with probability not less than $\Phi$ defined in~(\ref{probability2}).
Formula~(\ref{number1}) bounds the first term in the right-hand side
of~(\ref{basic_bound}).

To bound the second term in the right-hand side of~(\ref{basic_bound}),
we observe that $j \le k$, due to~(\ref{reduced_rank})
and the supposition of the theorem that $l \le m-k$.
Combining~(\ref{good_approx2}), (\ref{product2}),
(\ref{stretched_singular_value2}), and the fact that $j \le k$ yields
\begin{equation}
\label{number3}
\| Q \, S - (G \, (A \, A^\T)^i \, A)^\T \|
\le \sqrt{2 l} \; \gamma \; (\sigma_{j+1})^{2i+1}
  + \sqrt{2 \, \max(m-k,l)} \; \gamma \; (\sigma_{k+1})^{2i+1}
\end{equation}
with probability not less than $\Psi$ defined in~(\ref{probability32}).
Combining~(\ref{number2}) and~(\ref{number3}) yields
\begin{multline}
\label{number4}
\| F \| \, \| Q \, S - (G \, (A \, A^\T)^i \, A)^\T \|
\le \sqrt{2 \, l^2 \, \gamma^2 \, \beta^2} \; \sigma_{j+1} \\
  + \sqrt{2 \, l \, \max(m-k,l) \, \gamma^2 \, \beta^2
 \, \left(\frac{\sigma_{k+1}}{\sigma_j}\right)^{4i}} \;\; \sigma_{k+1}
\end{multline}
with probability not less than $\Pi$ defined in~(\ref{final_prob}).
The combination of Lemma~\ref{monotonicity}, (\ref{monotonicity_assump}),
and the fact that $j \le k$ justifies the use of $k$
(rather than the $j$ used in~(\ref{probability2}) for $\Phi$)
in the last term in the right-hand side of~(\ref{final_prob}).

Combining~(\ref{basic_bound}), (\ref{number1}), (\ref{number4}),
(\ref{reduced_rank}), (\ref{simplifying_assump}),
and the supposition of the theorem that $l \le m-k$
yields~(\ref{intermediate_step}), completing the proof.
\end{proof}

\begin{remark}
\label{par_remark}
Choosing~$l=k+12$, $\beta = 2.57$, and $\gamma = 2.43$ in~(\ref{final_prob})
and~(\ref{the_point}) yields
\begin{equation}
\label{explicit_eval}
\| A - U \, \Sigma \, V^\T \| \le 100 \, l
\, \left(\frac{m-k}{l}\right)^{1/(4i+2)} \, \sigma_{k+1}
\end{equation}
with probability greater than $1-10^{-15}$,
where $\sigma_{k+1}$ is the $(k+1)^\st$ greatest singular value of $A$.
Numerical experiments (some of which are reported in Section~\ref{numerical})
indicate that the factor $100 l$ in the right-hand side
of~(\ref{explicit_eval}) is much greater than necessary.
\end{remark}

\begin{remark}
\label{six-step}
Above, we permit $l$ to be any integer greater than $k$.
Stronger theoretical bounds on the accuracy are available when $l \ge 2k$.
Indeed, via an analysis similar to the proof of Theorem~\ref{the_theorem}
(using in addition the result stated in the abstract of~\cite{chen-dongarra}),
it can be shown that the following six-step algorithm with $l \ge 2k$
produces matrices $U$, $\Sigma$, and $V$ satisfying the bound~(\ref{the_point})
with its right-hand side reduced by a factor of $\sqrt{l}$:
\begin{enumerate}
\item[1.] Using a random number generator,
form a real $l \times m$ matrix $G$ whose entries are
i.i.d.\ Gaussian random variables of zero mean and unit variance,
and compute the $l \times n$ product matrix
\begin{equation}
\label{product2a}
R = G \, (A \, A^\T)^i \, A.
\end{equation}
\item[2.] Using a pivoted $QR$-decomposition algorithm,
form a real $n \times l$ matrix $Q$ whose columns are orthonormal,
such that there exists a real $l \times l$ matrix $S$ for which
\begin{equation}
\label{good_approx2a}
R^\T = Q \, S.
\end{equation}
(See, for example, Chapter~5 in~\cite{golub-van_loan} for details concerning
the construction of such a matrix $Q$.)
\item[3.] Compute the $m \times l$ product matrix
\begin{equation}
\label{product_ta}
T = A \, Q.
\end{equation}
\item[4.] Form an SVD of $T$,
\begin{equation}
\label{svd_smalla}
T = \tilde{U} \, \tilde{\Sigma} \, W^\T,
\end{equation}
where $\tilde{U}$ is a real $m \times l$ matrix whose columns are orthonormal,
$W$ is a real $l \times l$ matrix whose columns are orthonormal,
and $\tilde{\Sigma}$ is a real diagonal $l \times l$ matrix
whose only nonzero entries are nonnegative and appear in nonincreasing order
on the diagonal.
(See, for example, Chapter~8 in~\cite{golub-van_loan} for details
concerning the construction of such an SVD.)
\item[5.] Compute the $n \times l$ product matrix
\begin{equation}
\label{product3a}
\tilde{V} = Q \, W.
\end{equation}
\item[6.] Extract the leftmost $m \times k$ block $U$ of $\tilde{U}$,
the leftmost $n \times k$ block $V$ of $\tilde{V}$,
and the leftmost uppermost $k \times k$ block $\Sigma$ of $\tilde{\Sigma}$.
\end{enumerate}
\end{remark}

\subsection{Computational costs}
\label{costs}

In this subsection, we tabulate the number of floating-point operations
required by the five-step algorithm described
in Subsection~\ref{main_algorithm} as applied once to a matrix $A$.

The algorithm incurs the following costs
in order to compute an approximation to an SVD of $A$:
\begin{enumerate}
\item[1.] Forming $R$ in~(\ref{product2}) requires applying $A$
          to $il$ column vectors, and $A^\T$ to $(i+1) \, l$ column vectors.
\item[2.] Computing $Q$ in~(\ref{good_approx2})
          costs~$\bigoh(l^2 \, n)$.
\item[3.] Forming $T$ in~(\ref{product_t}) requires applying $A$
          to $k$ column vectors.
\item[4.] Computing the SVD~(\ref{svd_small}) of $T$ costs~$\bigoh(k^2 \, m)$.
\item[5.] Forming $V$ in~(\ref{product3}) costs~$\bigoh(k^2 \, n)$.
\end{enumerate}
Summing up the costs in Steps 1--5 above,
and using the fact that $k \le l \le m \le n$,
we conclude that the algorithm of Subsection~\ref{main_algorithm} costs
\begin{equation}
\label{svd_costs}
C_{\rm PCA} = (il+k) \cdot C_A + (il+l) \cdot C_{A^\tinyT} + \bigoh(l^2 \, n)
\end{equation}
floating-point operations,
where $C_A$ is the cost of applying $A$ to a real $n \times 1$ column vector,
and $C_{A^\tinyT}$ is the cost of applying $A^\T$
to a real $m \times 1$ column vector.

\begin{remark}
We observe that the algorithm
only requires applying $A$ to $il+k$ vectors and $A^\T$ to $il+l$ vectors;
it does not require explicit access to the individual entries of $A$.
This consideration can be important when $A$ and $A^\T$ are available solely
in the form of procedures for their applications to arbitrary vectors.
Often such procedures for applying $A$ and $A^\T$ cost much less than
the standard procedure for applying a dense matrix to a vector.
\end{remark}

\subsection{A modified algorithm}
\label{modified}

In this subsection, we describe a simple modification
of the algorithm described in Subsection~\ref{main_algorithm}.
Again, suppose that $i$, $k$, $l$, $m$, and $n$ are positive integers
with $k < l \le m-k$ and $m \le n$, and $A$ is a real $m \times n$ matrix.
Then, the following five-step algorithm constructs an approximation
to an SVD of $A^\T$ such that
\begin{equation}
\label{sort_of_svdmod}
\| A^\T - U \, \Sigma \, V^\T \| \le C \, m^{1/(4i)} \, \sigma_{k+1}
\end{equation}
with very high probability,
where $U$ is a real $n \times k$ matrix whose columns are orthonormal,
$V$ is a real $m \times k$ matrix whose columns are orthonormal,
$\Sigma$ is a real diagonal $k \times k$ matrix
whose entries are all nonnegative,
$\sigma_{k+1}$ is the $(k+1)^\st$ greatest singular value of $A$,
and $C$ is a constant independent of $A$ that depends on the parameters
of the algorithm:

\begin{enumerate}
\item[1.] Using a random number generator,
form a real $l \times m$ matrix $G$ whose entries are
i.i.d.\ Gaussian random variables of zero mean and unit variance,
and compute the $l \times m$ product matrix
\begin{equation}
\label{product2mod}
R = G \, (A \, A^\T)^i.
\end{equation}
\item[2.] Using an SVD,
form a real $m \times k$ matrix $Q$ whose columns are orthonormal,
such that there exists a real $k \times l$ matrix $S$ for which
\begin{equation}
\label{good_approx2mod}
\| Q \, S - R^\T \| \le \rho_{k+1},
\end{equation}
where $\rho_{k+1}$ is the $(k+1)^\st$ greatest singular value of $R$.
(See Observation~\ref{least_squares} for details concerning
the construction of such a matrix $Q$.)
\item[3.] Compute the $n \times k$ product matrix
\begin{equation}
\label{product_tmod}
T = A^\T \, Q.
\end{equation}
\item[4.] Form an SVD of $T$,
\begin{equation}
\label{svd_smallmod}
T = U \, \Sigma \, W^\T,
\end{equation}
where $U$ is a real $n \times k$ matrix whose columns are orthonormal,
$W$ is a real $k \times k$ matrix whose columns are orthonormal,
and $\Sigma$ is a real diagonal $k \times k$ matrix
whose entries are all nonnegative.
(See, for example, Chapter~8 in~\cite{golub-van_loan} for details
concerning the construction of such an SVD.)
\item[5.] Compute the $m \times k$ product matrix
\begin{equation}
\label{product3mod}
V = Q \, W.
\end{equation}
\end{enumerate}

Clearly, (\ref{sort_of_svdmod}) is similar to~(\ref{sort_of_svd}),
as~(\ref{product2mod}) is similar to~(\ref{product2}).

\begin{remark}
The ideas of Remark~\ref{six-step}
are obviously relevant to the algorithm of the present subsection, too.
\end{remark}

\subsection{Blanczos}
\label{blanczos}

In this subsection, we describe a modification of the algorithm
of Subsection~\ref{main_algorithm}, enhancing the accuracy
at a little extra computational expense.
Suppose that $i$, $k$, $l$, $m$, and $n$ are positive integers
with $k < l$ and $(i+1)l \le m-k$, and $A$ is a real $m \times n$ matrix,
such that $m \le n$.
Then, the following five-step algorithm constructs an approximation
$U \, \Sigma \, V^\T$ to an SVD of $A$:

\begin{enumerate}
\item[1.] Using a random number generator,
form a real $l \times m$ matrix $G$ whose entries are
i.i.d.\ Gaussian random variables of zero mean and unit variance,
and compute the $l \times n$ matrices
$R^{(0)}$, $R^{(1)}$, \dots, $R^{(i-1)}$, $R^{(i)}$
defined via the formulae
\begin{equation}
R^{(0)} = G \, A,
\end{equation}
\begin{equation}
R^{(1)} = R^{(0)} \, A^T \, A,
\end{equation}
\begin{equation}
R^{(2)} = R^{(1)} \, A^T \, A,
\end{equation}
\begin{equation*}
\vdots
\end{equation*}
\begin{equation}
R^{(i-1)} = R^{(i-2)} \, A^T \, A,
\end{equation}
\begin{equation}
R^{(i)} = R^{(i-1)} \, A^T \, A.
\end{equation}
Form the $((i+1)l) \times n$ matrix
\begin{equation}
\label{product23}
R = \left(\begin{array}{c} R^{(0)} \\ R^{(1)} \\ \vdots \\ R^{(i-1)} \\ R^{(i)}
\end{array}\right).
\end{equation}
\item[2.] Using a pivoted $QR$-decomposition algorithm,
form a real $n \times ((i+1)l)$ matrix $Q$ whose columns are orthonormal,
such that there exists a real $((i+1)l) \times ((i+1)l)$ matrix $S$ for which
\begin{equation}
\label{good_approx23}
R^\T = Q \, S.
\end{equation}
(See, for example, Chapter~5 in~\cite{golub-van_loan} for details concerning
the construction of such a matrix $Q$.)
\item[3.] Compute the $m \times ((i+1)l)$ product matrix
\begin{equation}
\label{product_t3}
T = A \, Q.
\end{equation}
\item[4.] Form an SVD of $T$,
\begin{equation}
\label{svd_small3}
T = U \, \Sigma \, W^\T,
\end{equation}
where $U$ is a real $m \times ((i+1)l)$ matrix whose columns are orthonormal,
$W$ is a real $((i+1)l) \times ((i+1)l)$ matrix whose columns are orthonormal,
and $\Sigma$ is a real diagonal $((i+1)l) \times ((i+1)l)$ matrix
whose entries are all nonnegative.
(See, for example, Chapter~8 in~\cite{golub-van_loan} for details
concerning the construction of such an SVD.)
\item[5.] Compute the $n \times ((i+1)l)$ product matrix
\begin{equation}
\label{product33}
V = Q \, W.
\end{equation}
\end{enumerate}

An analysis similar to the proof of Theorem~\ref{the_theorem} above
shows that the matrices $U$, $\Sigma$, and $V$ produced
by the algorithm of the present subsection satisfy
the same upper bounds~(\ref{the_point}) and~(\ref{explicit_eval})
as the matrices produced by the algorithm of Subsection~\ref{main_algorithm}.
If desired, one may produce a similarly accurate rank-$k$ approximation
by arranging $U$, $\Sigma$, and $V$ such that the diagonal entries
of $\Sigma$ appear in nonincreasing order,
and then discarding all but the leftmost $k$ columns of $U$
and all but the leftmost $k$ columns of $V$,
and retaining only the leftmost uppermost $k \times k$ block of $\Sigma$.
We will refer to the algorithm of the present subsection
as ``blanczos,'' due to its similarity with the block Lanczos method
(see, for example, Subsection~9.2.6 in~\cite{golub-van_loan}
for a description of the block Lanczos method).

\section{Numerical results}
\label{numerical}

In this section, we illustrate the performance of the algorithm
of the present paper via several numerical examples.

We use the algorithm to construct a rank-$k$ approximation,
with $k = 10$, to the $m \times (2m)$ matrix $A$ defined
via its singular value decomposition
\begin{equation}
\label{test_matrix}
A = U^{(A)} \, \Sigma^{(A)} \, (V^{(A)})^\T,
\end{equation}
where $U^{(A)}$ is an $m \times m$ Hadamard matrix
(a unitary matrix whose entries are all $\pm 1/\sqrt{m}$),
$V^{(A)}$ is a $(2m) \times (2m)$ Hadamard matrix,
and $\Sigma^{(A)}$ is an $m \times (2m)$ matrix
whose entries are zero off the main diagonal,
and whose diagonal entries are defined
in terms of the $(k+1)^\st$ singular value $\sigma_{k+1}$ via the formulae
\begin{equation}
\Sigma^{(A)}_{j,j} = \sigma_j = (\sigma_{k+1})^{\lfloor j/2 \rfloor/5}
\end{equation}
for $j = 1$,~$2$, \dots, $9$,~$10$,
where $\lfloor j/2 \rfloor$ is the greatest integer less than
or equal to $j/2$, and
\begin{equation}
\Sigma^{(A)}_{j,j} = \sigma_j = \sigma_{k+1} \cdot \frac{m-j}{m-11}
\end{equation}
for $j = 11$,~$12$, \dots, $m-1$,~$m$.
Thus, $\sigma_1 = 1$ and $\sigma_k = \sigma_{k+1}$ (recall that $k = 10$).
We always choose $\sigma_{k+1} < 1$,
so that $\sigma_1 \ge \sigma_2 \ge \dots \ge \sigma_{m-1} \ge \sigma_m$.

Figure~1 plots the singular values
$\sigma_1$,~$\sigma_2$, \dots, $\sigma_{m-1}$,~$\sigma_m$
of $A$ with $m = 512$ and $\sigma_{k+1} = .001$;
these parameters correspond to the first row of numbers in Table~1,
the first row of numbers in Table~2, and the first row of numbers in Table~6.

Table~1 reports the results of applying the five-step algorithm
of Subsection~\ref{main_algorithm} to matrices of various sizes, with $i = 1$.
Table~2 reports the results of applying the five-step algorithm
of Subsection~\ref{main_algorithm} to matrices of various sizes, with $i = 0$.
The algorithms of~\cite{sarlos3}, \cite{sarlos4},
and~\cite{liberty-woolfe-martinsson-rokhlin-tygert}
for low-rank approximation are essentially the same as the algorithm used
for Table~2 (with $i=0$).

Table~3 reports the results of applying the five-step algorithms
of Subsections~\ref{main_algorithm} and~\ref{modified}
with varying numbers of iterations $i$.
Rows in the table where $i$ is enclosed in parentheses correspond
to the algorithm of Subsection~\ref{modified};
rows where $i$ is not enclosed in parentheses correspond
to the algorithm of Subsection~\ref{main_algorithm}.

Table~4 reports the results of applying the five-step algorithm
of Subsection~\ref{main_algorithm} to matrices
whose best rank-$k$ approximations have varying accuracies.
Table~5 reports the results of applying the blanczos algorithm
of Subsection~\ref{blanczos} to matrices
whose best rank-$k$ approximations have varying accuracies.

Table~6 reports the results of calculating pivoted $QR$-decompositions,
via plane (Householder) reflections, of matrices of various sizes.
We computed the pivoted $QR$-decomposition of the transpose of $A$ defined
in~(\ref{test_matrix}), rather than of $A$ itself, for reasons of accuracy
and efficiency. As pivoted $QR$-decomposition requires dense matrix arithmetic,
our 1~GB of random-access memory (RAM) imposed the limit $m \le 4096$
for Table~6.

The headings of the tables have the following meanings:
\begin{itemize}
\item $m$ is the number of rows in $A$, the matrix being approximated.
\item $n$ is the number of columns in $A$, the matrix being approximated.
\item $i$ is the integer parameter used in the algorithms
      of Subsections~\ref{main_algorithm}, \ref{modified}, and~\ref{blanczos}.
      Rows in the tables where $i$ is enclosed in parentheses correspond
      to the algorithm of Subsection~\ref{modified};
      rows where $i$ is not enclosed in parentheses correspond
      to either the algorithm of Subsection~\ref{main_algorithm} or
      that of Subsection~\ref{blanczos}.
\item $t$ is the time in seconds required by the algorithm to create
      an approximation and compute its accuracy $\delta$.
\item $\sigma_{k+1}$ is the $(k+1)^\st$ greatest singular value of $A$,
      the matrix being approximated; $\sigma_{k+1}$ is also the accuracy
      of the best possible rank-$k$ approximation to $A$.
\item $\delta$ is the accuracy of the approximation $U \, \Sigma \, V^\T$
      (or $(QRP)^\T$, for Table~6) constructed by the algorithm.
      For Tables~1--5,
\begin{equation}
\delta = \| A - U \, \Sigma \, V^\T \|,
\end{equation}
where $U$ is an $m \times k$ matrix whose columns are orthonormal,
$V$ is an $n \times k$ matrix whose columns are orthonormal,
and $\Sigma$ is a diagonal $k \times k$ matrix whose entries
are all nonnegative; for Table~6,
\begin{equation}
\delta = \| A - (QRP)^\T \|,
\end{equation}
where $P$ is an $m \times m$ permutation matrix,
$R$ is a $k \times m$ upper-triangular (meaning upper-trapezoidal) matrix,
and $Q$ is an $n \times k$ matrix whose columns are orthonormal.
\end{itemize}

The values for $t$ are the average values over 3 independent randomized trials
of the algorithm. The values for $\delta$ are the worst (maximum) values
encountered in 3 independent randomized trials of the algorithm.
The values for $\delta$ in each trial are those produced by 20 iterations
of the power method applied to $A - U \, \Sigma \, V^\T$
(or $A - (QRP)^\T$, for Table~6),
started with a vector whose entries
are i.i.d.\ centered Gaussian random variables.
The theorems of~\cite{dixon} and~\cite{kuczynski-wozniakowski}
guarantee that this power method produces accurate results
with overwhelmingly high probability.

We performed all computations using IEEE standard double-precision variables,
whose mantissas have approximately one bit of precision less than 16 digits
(so that the relative precision of the variables is approximately .2E--15).
We ran all computations on one core
of a 1.86~GHz Intel Centrino Core Duo microprocessor
with 2~MB of L2 cache and 1~GB of RAM.
We compiled the Fortran~77 code
using the Lahey/Fujitsu Linux Express v6.2 compiler,
with the optimization flag {\tt {-}{-}o2} enabled.
We implemented a fast Walsh-Hadamard transform
to apply rapidly the Hadamard matrices $U^{(A)}$ and $V^{(A)}$
in~(\ref{test_matrix}).
We used plane (Householder) reflections
to compute all pivoted $QR$-decompositions.
We used the LAPACK 3.1.1 divide-and-conquer SVD routine {\tt dgesdd}
to compute all full SVDs.
For the parameter $l$, we set $l = 12$ $(= k+2)$
for all of the examples reported here.

The experiments reported here and our further tests point
to the following:

\begin{enumerate}
\item The accuracies in Table~1 are superior to those in Table~2;
the algorithm performs much better with $i>0$.
(The algorithms of~\cite{liberty-woolfe-martinsson-rokhlin-tygert},
\cite{sarlos3}, and~\cite{sarlos4}
for low-rank approximation are essentially the same as the algorithm used
for Tables~1 and~2 when $i=0$.)
\item The accuracies in Table~1 are superior to the corresponding accuracies
in Table~6; the algorithm of the present paper produces higher accuracy
than the classical pivoted $QR$-decompositions for matrices whose spectra
decay slowly (such as those matrices tested in the present section).
\item The accuracies in Tables~1--3 appear to be proportional
to $m^{1/(4i+2)} \, \sigma_{k+1}$ for the algorithm
of Subsection~\ref{main_algorithm},
and to be proportional to $m^{1/(4i)} \, \sigma_{k+1}$ for the algorithm
of Subsection~\ref{modified},
in accordance with~(\ref{sort_of_svd}) and~(\ref{sort_of_svdmod}).
The numerical results reported here, as well as our further experiments,
indicate that the theoretical bound~(\ref{the_point}) on the accuracy
should remain valid with a greatly reduced constant in the right-hand side,
independent of the matrix $A$ being approximated.
See item~6 below for a discussion of Tables~4 and~5.
\item The timings in Tables~1--5 are consistent with~(\ref{svd_costs}),
as we could (and did) apply the Hadamard matrices $U^{(A)}$ and $V^{(A)}$
in~(\ref{test_matrix}) to vectors via fast Walsh-Hadamard transforms
at a cost of $\bigoh(m \, \log(m))$ floating-point operations
per matrix-vector multiplication.
\item The quality of the pseudorandom number generator has almost no effect
on the accuracy of the algorithm, nor does substituting uniform variates
for the normal variates.
\item The accuracies in Table~5 are superior to those in Table~4,
particularly when the $k^\th$ greatest singular value $\sigma_k$
of the matrix $A$ being approximated is very small. Understandably,
the algorithm of Subsection~\ref{main_algorithm} would seem to break down
when $(\sigma_k)^{2i+1}$ is less than the machine precision,
while $\sigma_k$ itself is not,
unlike the blanczos algorithm of Subsection~\ref{blanczos}.
When $(\sigma_k)^{2i+1}$ is much less than the machine precision,
while $\sigma_k$ is not,
the accuracy of blanczos in the presence of roundoff is similar to that
of the algorithm of Subsection~\ref{main_algorithm} run with a reduced $i$.
When $(\sigma_k)^{2i+1}$ is much greater than the machine precision,
the accuracy of blanczos is similar to that of the algorithm
of Subsection~\ref{main_algorithm} run with $i$ being the same as
in the blanczos algorithm.
Since the blanczos algorithm of Subsection~\ref{blanczos}
is so tolerant of roundoff,
we suspect that the blanczos algorithm is
a better general-purpose black-box tool
for the computation of principal component analyses,
despite its somewhat higher cost as compared with the algorithms
of Subsections~\ref{main_algorithm} and~\ref{modified}.
\end{enumerate}

\begin{remark}
A MATLAB\registered\ implementation of the blanczos algorithm
of Subsection~\ref{blanczos} is available on the file exchange at
{\tt http://www.mathworks.com} in the package entitled,
``Principal Component Analysis.''
\end{remark}

\section{Appendix}
\label{appendix}

In this appendix, we restate and prove Lemmas~\ref{all_together2}
and~\ref{probability_bounds2} from Section~\ref{apparatus}.

The following lemma, stated earlier as Lemma~\ref{all_together2}
in Section~\ref{apparatus},
shows that the product $A \, Q \, Q^\T$
of matrices $A$, $Q$, and $Q^\T$
is a good approximation to a matrix $A$,
provided that there exist matrices $G$ and $S$ such that
\begin{enumerate}
\item[1.] the columns of $Q$ are orthonormal,
\item[2.] $Q \, S$ is a good approximation to $(G \, (A \, A^\T)^i \, A)^\T$,
and
\item[3.] there exists a matrix $F$ such that $\| F \|$ is not too large,
and $F \, G \, (A \, A^\T)^i \, A$ is a good approximation to $A$.
\end{enumerate}

\begin{lemma}
\label{all_together22}
Suppose that $i$, $k$, $l$, $m$, and~$n$ are positive integers
with $k \le l \le m \le n$.
Suppose further that $A$ is a real $m \times n$ matrix,
$Q$ is a real $n \times k$ matrix whose columns are orthonormal,
$S$ is a real $k \times l$ matrix,
$F$ is a real $m \times l$ matrix,
and $G$ is a real $l \times m$ matrix.

Then,
\begin{equation}
\label{reconstruction22}
\| A \, Q \, Q^\T - A \|
\le 2 \, \| F \, G \, (A \, A^\T)^i \, A - A \|
  + 2 \, \| F \| \, \| Q \, S - (G \, (A \, A^\T)^i \, A)^\T \|.
\end{equation}
\end{lemma}

\begin{proof}
The proof is straightforward, but tedious, as follows.

To simplify notation, we define
\begin{equation}
\label{shorter}
B = (A \, A^\T)^i \, A.
\end{equation}

We obtain from the triangle inequality that
\begin{multline}
\label{triangle}
\| A \, Q \, Q^\T - A \|
\le \| A \, Q \, Q^\T - F \, G \, B \, Q \, Q^\T \|
  + \| F \, G \, B \, Q \, Q^\T - F \, G \, B \| \\
  + \| F \, G \, B - A \|.
\end{multline}

First, we provide a bound
for $\| A \, Q \, Q^\T - F \, G \, B \, Q \, Q^\T \|$.
Clearly,
\begin{equation}
\label{bound0}
\| A \, Q \, Q^\T - F \, G \, B \, Q \, Q^\T \|
\le \| A - F \, G \, B \| \, \| Q \| \, \| Q^\T \|.
\end{equation}
It follows from the fact that the columns of $Q$ are orthonormal that
\begin{equation}
\label{bound1}
\| Q \| \le 1
\end{equation}
and
\begin{equation}
\label{bound2}
\| Q^\T \| \le 1.
\end{equation}
Combining~(\ref{bound0}), (\ref{bound1}), and~(\ref{bound2}) yields
\begin{equation}
\label{simpler}
\| A \, Q \, Q^\T - F \, G \, B \, Q \, Q^\T \| \le \| A - F \, G \, B \|.
\end{equation}

Next, we provide a bound
for $\| F \, G \, B \, Q \, Q^\T - F \, G \, B \|$.
Clearly,
\begin{equation}
\label{triangle4}
\| F \, G \, B \, Q \, Q^\T - F \, G \, B \|
\le \| F \| \, \| G \, B \, Q \, Q^\T - G \, B \|.
\end{equation}
It follows from the triangle inequality that
\begin{multline}
\label{triangle3}
\| G \, B \, Q \, Q^\T - G \, B \|
\le \| G \, B \, Q \, Q^\T - S^\T \, Q^\T \, Q \, Q^\T \| \\
  + \| S^\T \, Q^\T \, Q \, Q^\T - S^\T \, Q^\T \|
  + \| S^\T \, Q^\T - G \, B \|.
\end{multline}

Furthermore,
\begin{equation}
\label{prev}
\| G \, B \, Q \, Q^\T - S^\T \, Q^\T \, Q \, Q^\T \|
\le \| G \, B - S^\T \, Q^\T \| \, \| Q \| \, \| Q^\T \|.
\end{equation}
Combining~(\ref{prev}), (\ref{bound1}), and~(\ref{bound2}) yields
\begin{equation}
\label{bound4}
\| G \, B \, Q \, Q^\T - S^\T \, Q^\T \, Q \, Q^\T \|
\le \| G \, B - S^\T \, Q^\T \|.
\end{equation}

Also, it follows from the fact that the columns of $Q$ are orthonormal that
\begin{equation}
\label{orthonormal}
Q^\T \, Q = \Id.
\end{equation}
It follows from~(\ref{orthonormal}) that
\begin{equation}
\label{vanish}
\| S^\T \, Q^\T \, Q \, Q^\T - S^\T \, Q^\T \| = 0.
\end{equation}

Combining~(\ref{triangle3}), (\ref{bound4}), and~(\ref{vanish}) yields
\begin{equation}
\label{triangle5}
\| G \, B \, Q \, Q^\T - G \, B \| \le 2 \, \| S^\T \, Q^\T - G \, B \|.
\end{equation}
Combining~(\ref{triangle4}) and~(\ref{triangle5}) yields
\begin{equation}
\label{triangle6}
\| F \, G \, B \, Q \, Q^\T - F \, G \, B \|
\le 2 \, \| F \| \, \| S^\T \, Q^\T - G \, B \|.
\end{equation}

Combining~(\ref{triangle}), (\ref{simpler}), (\ref{triangle6}),
and~(\ref{shorter}) yields~(\ref{reconstruction22}).
\end{proof}

The following lemma, stated earlier as Lemma~\ref{probability_bounds2}
in Section~\ref{apparatus}, shows that,
for any positive integer $i$, matrix $A$, and matrix $G$ whose entries are
i.i.d.\ Gaussian random variables of zero mean and unit variance,
with very high probability there exists a matrix $F$
with a reasonably small norm,
such that $F \, G \, (A \, A^\T)^i \, A$ is a good approximation to $A$.
This lemma is similar to Lemma~19 of~\cite{martinsson-rokhlin-tygert3}.

\begin{lemma}
\label{probability_bounds22}
Suppose that $i$, $j$, $k$, $l$, $m$, and~$n$ are positive integers
with $j < k < l < m \le n$.
Suppose further that $A$ is a real $m \times n$ matrix,
$G$ is a real $l \times m$ matrix whose entries are
i.i.d.\ Gaussian random variables of zero mean and unit variance,
and $\beta$ and $\gamma$ are positive real numbers, such that
the $j^\ith$ greatest singular value $\sigma_j$ of $A$ is positive,
$\gamma > 1$, and
\begin{multline}
\label{probability22}
\Phi
  = 1 - \frac{1}{\sqrt{2 \pi \, (l-j+1)}}
 \, \left( \frac{e}{(l-j+1) \, \beta} \right)^{l-j+1} \\
  - \frac{1}{4 \, (\gamma^2-1) \, \sqrt{\pi \, \max(m-k,l) \; \gamma^2}}
    \left( \frac{2 \gamma^2}{e^{\gamma^2-1}} \right)^{\max(m-k,\,l)} \\
  - \frac{1}{4 \, (\gamma^2-1) \, \sqrt{\pi \, l \, \gamma^2}}
    \left( \frac{2 \gamma^2}{e^{\gamma^2-1}} \right)^l
\end{multline}
is nonnegative.

Then, there exists a real $m \times l$ matrix $F$ such that
\begin{multline}
\label{approximation22}
\| F \, G \, (A \, A^\T)^i \, A - A \|
\le \sqrt{ 2 l^2 \, \beta^2 \, \gamma^2 + 1 }
 \;\; \sigma_{j+1} \\
  + \sqrt{ 2 l \, \max(m-k,l) \, \beta^2 \, \gamma^2
        \, \left( \frac{\sigma_{k+1}}{\sigma_j} \right)^{4i} + 1 }
 \;\; \sigma_{k+1}
\end{multline}
and
\begin{equation}
\label{small_norm22}
\| F \| \le \frac{\sqrt{l} \; \beta}{(\sigma_j)^{2i}}
\end{equation}
with probability not less than $\Phi$ defined in~(\ref{probability22}),
where $\sigma_j$ is the $j^\ith$ greatest singular value of $A$,
$\sigma_{j+1}$ is the $(j+1)^\ist$ greatest singular value of $A$,
and $\sigma_{k+1}$ is the $(k+1)^\ist$ greatest singular value of $A$.
\end{lemma}

\begin{proof}
We prove the existence of a matrix $F$ satisfying~(\ref{approximation22})
and~(\ref{small_norm22}) by constructing one.

We start by forming an SVD of $A$,
\begin{equation}
\label{svd2}
A = U \, \Sigma \, V^\T,
\end{equation}
where $U$ is a real unitary $m \times m$ matrix,
$\Sigma$ is a real diagonal $m \times m$ matrix,
and $V$ is a real $n \times m$ matrix whose columns are orthonormal, such that
\begin{equation}
\label{ordering2}
\Sigma_{p,p} = \sigma_p
\end{equation}
for $p = 1$,~$2$, \dots, $m-1$,~$m$,
where $\Sigma_{p,p}$ is the entry in row $p$ and column $p$ of $\Sigma$,
and $\sigma_p$ is the $p^\th$ greatest singular value of $A$.

Next, we define auxiliary matrices
$H$, $R$, $\Gamma$, $S$, $T$, $\Theta$, and $P$.
We define $H$ to be the leftmost $l \times j$ block
of the $l \times m$ matrix $G \, U$,
$R$ to be the $l \times (k-j)$ block of $G \, U$
whose first column is the $(k+1)^\st$ column of $G \, U$,
and $\Gamma$ to be the rightmost $l \times (m-k)$ block
of $G \, U$, so that
\begin{equation}
\label{partition2}
G \, U = \left( \begin{array}{c|c|c} H & R & \Gamma \end{array} \right).
\end{equation}
Combining the fact that $U$ is real and unitary,
and the fact that the entries of $G$ are i.i.d.\ Gaussian
random variables of zero mean and unit variance,
we see that the entries of $H$ are also i.i.d.\ Gaussian
random variables of zero mean and unit variance,
as are the entries of $R$, and as are the entries of $\Gamma$.
We define $H^{(-1)}$ to be the real $j \times l$ matrix
given by the formula
\begin{equation}
\label{definition_of_pseudoinverse2}
H^{(-1)} = (H^\T \, H)^{-1} \, H^\T
\end{equation}
($H^\T \, H$ is invertible with high probability
due to Lemma~\ref{least_value}).
We define $S$ to be the leftmost uppermost $j \times j$ block of $\Sigma$,
$T$ to be the $(k-j) \times (k-j)$ block of $\Sigma$
whose leftmost uppermost entry is the entry
in the $(j+1)^\st$ row and $(j+1)^\st$ column of $\Sigma$,
and $\Theta$ to be the rightmost lowermost $(m-k) \times (m-k)$ block
of $\Sigma$, so that
\begin{equation}
\label{svd_partition2}
\Sigma
= \left( \begin{array}{c|c|c} S   & \0s & \0s    \\\hline
                              \0s & T   & \0s    \\\hline
                              \0s & \0s & \Theta
         \end{array} \right).
\end{equation}
We define $P$ to be the real $m \times l$ matrix
whose uppermost $j \times l$ block is the product $S^{-2i} \, H^{(-1)}$,
whose entries are zero in the $(k-j) \times l$ block whose first row
is the $(j+1)^\st$ row of $P$,
and whose entries in the lowermost $(m-k) \times l$ block are zero,
so that
\begin{equation}
\label{pad2}
P = \left( \begin{array}{c} S^{-2i} \, H^{(-1)} \\\hline \0s
                                                \\\hline \0s
           \end{array} \right).
\end{equation}

Finally, we define $F$ to be the $m \times l$ matrix given by
\begin{equation}
\label{inverter2}
F = U \, P = U \, \left( \begin{array}{c} S^{-2i} \, H^{(-1)} \\\hline
                                          \0s \\\hline \0s
                         \end{array} \right).
\end{equation}

Combining~(\ref{definition_of_pseudoinverse2}), (\ref{pseudoinverse_norm}), 
the fact that the entries of $H$ are i.i.d.\ Gaussian
random variables of zero mean and unit variance,
and Lemma~\ref{least_value} yields
\begin{equation}
\label{pseudoinverse2}
\left\| H^{(-1)} \right\| \le \sqrt{l} \; \beta
\end{equation}
with probability not less than
\begin{equation}
1 - \frac{1}{\sqrt{2 \pi \, (l-j+1)}}
 \, \left( \frac{e}{(l-j+1) \, \beta} \right)^{l-j+1}.
\end{equation}
Combining~(\ref{inverter2}), (\ref{pseudoinverse2}), (\ref{svd_partition2}),
(\ref{ordering2}), the fact that $\Sigma$ is zero off its main diagonal,
and the fact that $U$ is unitary yields~(\ref{small_norm22}).

We now show that $F$ defined in~(\ref{inverter2})
satisfies~(\ref{approximation22}).

Combining~(\ref{svd2}), (\ref{partition2}), and~(\ref{inverter2}) yields
\begin{equation}
\label{simplification12}
F \, G \, (A \, A^\T)^i \, A - A
= U \, \left( \left( \begin{array}{c} S^{-2i} \, H^{(-1)} \\\hline
                                      \0s \\\hline \0s
                     \end{array} \right)
              \left( \begin{array}{c|c|c} H & R & \Gamma \end{array} \right)
              \, \Sigma^{2i}
            - \Id \right) \, \Sigma \, V^\T.
\end{equation}
Combining~(\ref{definition_of_pseudoinverse2})
and~(\ref{svd_partition2}) yields
\begin{multline}
\label{simplification22}
\left( \left( \begin{array}{c} S^{-2i} \, H^{(-1)} \\\hline \0s
                                                   \\\hline \0s
              \end{array} \right)
       \left( \begin{array}{c|c|c} H & R & \Gamma \end{array} \right)
       \, \Sigma^{2i}
     - \Id \right) \, \Sigma \\
= \left( \begin{array}{c|c|c}
         \0s & S^{-2i} \, H^{(-1)} \, R \; T^{2i+1} &
               S^{-2i} \, H^{(-1)} \, \Gamma \, \Theta^{2i+1} \\\hline
         \0s & -T & \0s \\\hline
         \0s & \0s & -\Theta
  \end{array} \right).
\end{multline}
Furthermore,
\begin{multline}
\label{Frobenius2}
\left\| \left( \begin{array}{c|c|c}
       \0s & S^{-2i} \, H^{(-1)} \, R \; T^{2i+1} &
             S^{-2i} \, H^{(-1)} \, \Gamma \, \Theta^{2i+1} \\\hline
       \0s & -T & \0s \\\hline
       \0s & \0s & -\Theta
\end{array} \right) \right\|^2 \\
\le \left\| S^{-2i} \, H^{(-1)} \, R \, T^{2i+1} \right\|^2
  + \left\| S^{-2i} \, H^{(-1)} \, \Gamma \, \Theta^{2i+1} \right\|^2
  + \| T \|^2 + \| \Theta \|^2.
\end{multline}

Moreover,
\begin{equation}
\label{product_of_norms2}
\left\| S^{-2i} \, H^{(-1)} \, R \, T^{2i+1} \right\|
\le \left\| S^{-1} \right\|^{2i} \, \left\| H^{(-1)} \right\|
 \, \| R \| \, \| T \|^{2i+1}
\end{equation}
and
\begin{equation}
\label{product_of_norms3}
\left\| S^{-2i} \, H^{(-1)} \, \Gamma \, \Theta^{2i+1} \right\|
\le \left\| S^{-1} \right\|^{2i} \, \left\| H^{(-1)} \right\|
 \, \| \Gamma \| \, \| \Theta \|^{2i+1}.
\end{equation}
Combining~(\ref{svd_partition2}) and~(\ref{ordering2}) yields
\begin{equation}
\label{singular_value_bound1}
\left\| S^{-1} \right\| \le \frac{1}{\sigma_j},
\end{equation}
\begin{equation}
\label{singular_value_bound2}
\| T \| \le \sigma_{j+1},
\end{equation}
and
\begin{equation}
\label{singular_value_bound3}
\| \Theta \| \le \sigma_{k+1}.
\end{equation}
Combining~(\ref{simplification12})--(\ref{singular_value_bound3})
and the fact that the columns of $U$ are orthonormal,
as are the columns of $V$, yields
\begin{multline}
\label{almost_there2}
\| F \, G \, (A \, A^\T)^i \, A - A \|^2
\le \left( \left\| H^{(-1)} \right\|^2 \, \| R \|^2
        \, \left( \frac{\sigma_{j+1}}{\sigma_j} \right)^{4i} + 1 \right)
 \, (\sigma_{j+1})^2 \\
  + \left( \left\| H^{(-1)} \right\|^2 \, \| \Gamma \|^2
        \, \left( \frac{\sigma_{k+1}}{\sigma_j} \right)^{4i} + 1 \right)
 \, (\sigma_{k+1})^2.
\end{multline}

Combining Lemma~\ref{greatest_value}
and the fact that the entries of $R$ are
i.i.d.\ Gaussian random variables of zero mean and unit variance,
as are the entries of $\Gamma$, yields
\begin{equation}
\label{residual2}
\| R \| \le \sqrt{2l} \; \gamma
\end{equation}
and
\begin{equation}
\label{residual3}
\| \Gamma \| \le \sqrt{2 \, \max(m-k,l)} \; \gamma,
\end{equation}
with probability not less than
\begin{multline}
1 - \frac{1}{4 \, (\gamma^2-1) \, \sqrt{\pi \, \max(m-k,l) \, \gamma^2}}
    \left( \frac{2 \gamma^2}{e^{\gamma^2-1}} \right)^{\max(m-k,\,l)} \\
  - \frac{1}{4 \, (\gamma^2-1) \, \sqrt{\pi \, l \, \gamma^2}}
    \left( \frac{2 \gamma^2}{e^{\gamma^2-1}} \right)^l.
\end{multline}
Combining~(\ref{almost_there2}), (\ref{pseudoinverse2}),
(\ref{residual2}), and~(\ref{residual3}) yields
\begin{multline}
\label{pre-approximation2}
\| F \, G \, (A \, A^\T)^i \, A - A \|^2
\le \left( 2 l^2 \, \beta^2 \, \gamma^2
        \, \left( \frac{\sigma_{j+1}}{\sigma_j} \right)^{4i} + 1 \right)
 \, (\sigma_{j+1})^2 \\
  + \left( 2 l \, \max(m-k,l) \, \beta^2 \, \gamma^2
        \, \left( \frac{\sigma_{k+1}}{\sigma_j} \right)^{4i} + 1 \right)
 \, (\sigma_{k+1})^2
\end{multline}
with probability not less than $\Phi$ defined in~(\ref{probability22}).
Combining~(\ref{pre-approximation2}),
the fact that $\sigma_{j+1} \le \sigma_j$, and the fact that
\begin{equation}
\sqrt{x + y} \le \sqrt{x} + \sqrt{y}
\end{equation}
for any nonnegative real numbers $x$ and $y$
yields~(\ref{approximation22}). 
\end{proof}

\section*{Acknowledgements}
We thank Ming Gu for suggesting the combination
of the Lanczos method with randomized methods
for the low-rank approximation of matrices.
We are grateful for many helpful discussions
with R. Raphael Coifman and Yoel Shkolnisky.
We thank the anonymous referees for their useful suggestions.

\begin{figure}[b]
\begin{center}
\begin{tabular}{r|r|c|r|r|r}
   $m$ &     $n$ & $i$ &      $t$ & $\sigma_{k+1}$ & $\delta$ \\\hline
                                                                \hline
   512 &    1024 &   1 & .13E--01 &           .001 &    .0011 \\\hline
  2048 &    4096 &   1 & .56E--01 &           .001 &    .0013 \\\hline
  8192 &   16384 &   1 & .25E--00 &           .001 &    .0018 \\\hline
 32768 &   65536 &   1 &  .12E+01 &           .001 &    .0024 \\\hline
131072 &  262144 &   1 &  .75E+01 &           .001 &    .0037 \\\hline
524288 & 1048576 &   1 &  .36E+02 &           .001 &    .0039 \\\hline
\end{tabular}
\\\vspace{.125in}
Table~1: Five-step algorithm of Subsection~\ref{main_algorithm}
\end{center}
\end{figure}

\begin{figure}
\begin{center}
\begin{tabular}{r|r|c|r|r|r}
    $m$ &    $n$ & $i$ &      $t$ & $\sigma_{k+1}$ & $\delta$ \\\hline
                                                                \hline
   512 &    1024 &   0 & .14E--01 &           .001 &     .012 \\\hline
  2048 &    4096 &   0 & .47E--01 &           .001 &     .027 \\\hline
  8192 &   16384 &   0 & .22E--00 &           .001 &     .039 \\\hline
 32768 &   65536 &   0 &  .10E+01 &           .001 &     .053 \\\hline
131072 &  262144 &   0 &  .60E+01 &           .001 &     .110 \\\hline
524288 & 1048576 &   0 &  .29E+02 &           .001 &     .220 \\\hline
\end{tabular}
\\\vspace{.125in}
Table~2: Five-step algorithm of Subsection~\ref{main_algorithm}
\end{center}
\end{figure}

\begin{figure}
\begin{center}
\begin{tabular}{r|r|c|r|r|r}
   $m$ &     $n$ & $i$ &     $t$ & $\sigma_{k+1}$ & $\delta$ \\\hline
                                                               \hline
524288 & 1048576 &   0 & .29E+02 &            .01 &     .862 \\\hline
524288 & 1048576 & (1) & .31E+02 &            .01 &     .091 \\\hline
524288 & 1048576 &   1 & .36E+02 &            .01 &     .037 \\\hline
524288 & 1048576 & (2) & .38E+02 &            .01 &     .025 \\\hline
524288 & 1048576 &   2 & .43E+02 &            .01 &     .022 \\\hline
524288 & 1048576 & (3) & .45E+02 &            .01 &     .015 \\\hline
524288 & 1048576 &   3 & .49E+02 &            .01 &     .010 \\\hline
\end{tabular}
\\\vspace{.125in}
Table~3: Five-step algorithms of Subsections~\ref{main_algorithm}
         and~\ref{modified} \\\quad\quad\quad\quad
         (parentheses around $i$ designate Subsection~\ref{modified})
\end{center}
\end{figure}

\begin{figure}
\begin{center}
\begin{tabular}{r|r|c|r|r|r}
   $m$ &    $n$ & $i$ &     $t$ & $\sigma_{k+1}$ & $\delta$ \\\hline
                                                              \hline
262144 & 524288 &   1 & .17E+02 &       .10E--02 & .39E--02 \\\hline
262144 & 524288 &   1 & .17E+02 &       .10E--04 & .10E--03 \\\hline
262144 & 524288 &   1 & .17E+02 &       .10E--06 & .25E--05 \\\hline
262144 & 524288 &   1 & .17E+02 &       .10E--08 & .90E--06 \\\hline
262144 & 524288 &   1 & .17E+02 &       .10E--10 & .55E--07 \\\hline
262144 & 524288 &   1 & .17E+02 &       .10E--12 & .51E--08 \\\hline
262144 & 524288 &   1 & .17E+02 &       .10E--14 & .10E--05 \\\hline
\end{tabular}
\\\vspace{.125in}
Table~4: Five-step algorithm of Subsection~\ref{main_algorithm}
\end{center}
\end{figure}

\begin{figure}
\begin{center}
\begin{tabular}{r|r|c|r|r|r}
   $m$ &    $n$ & $i$ &     $t$ & $\sigma_{k+1}$ &   $\delta$ \\\hline
                                                                \hline
262144 & 524288 &   1 & .31E+02 &       .10E--02 &   .35E--02 \\\hline
262144 & 524288 &   1 & .31E+02 &       .10E--04 &   .15E--04 \\\hline
262144 & 524288 &   1 & .31E+02 &       .10E--06 &   .24E--05 \\\hline
262144 & 524288 &   1 & .31E+02 &       .10E--08 &   .11E--06 \\\hline
262144 & 524288 &   1 & .31E+02 &       .10E--10 &   .19E--08 \\\hline
262144 & 524288 &   1 & .31E+02 &       .10E--12 &   .25E--10 \\\hline
262144 & 524288 &   1 & .31E+02 &       .10E--14 &   .53E--11 \\\hline
\end{tabular}
\\\vspace{.125in}
Table~5: Five-step algorithm of Subsection~\ref{blanczos}
\end{center}
\end{figure}

\begin{figure}
\begin{center}
\begin{tabular}{r|r|r|r|r}
 $m$ &  $n$ &      $t$ & $\sigma_{k+1}$ & $\delta$ \\\hline
                                                     \hline
 512 & 1024 & .60E--01 &           .001 &    .0047 \\\hline
1024 & 2048 & .29E--00 &           .001 &    .0065 \\\hline
2048 & 4096 &  .11E+01 &           .001 &    .0092 \\\hline
4096 & 8192 &  .43E+01 &           .001 &    .0131 \\\hline
\end{tabular}
\\\vspace{.125in}
Table~6: Pivoted $QR$-decomposition
\end{center}
\end{figure}

\begin{figure}
\begin{center}
\rotatebox{-90}{\scalebox{.28}{\includegraphics{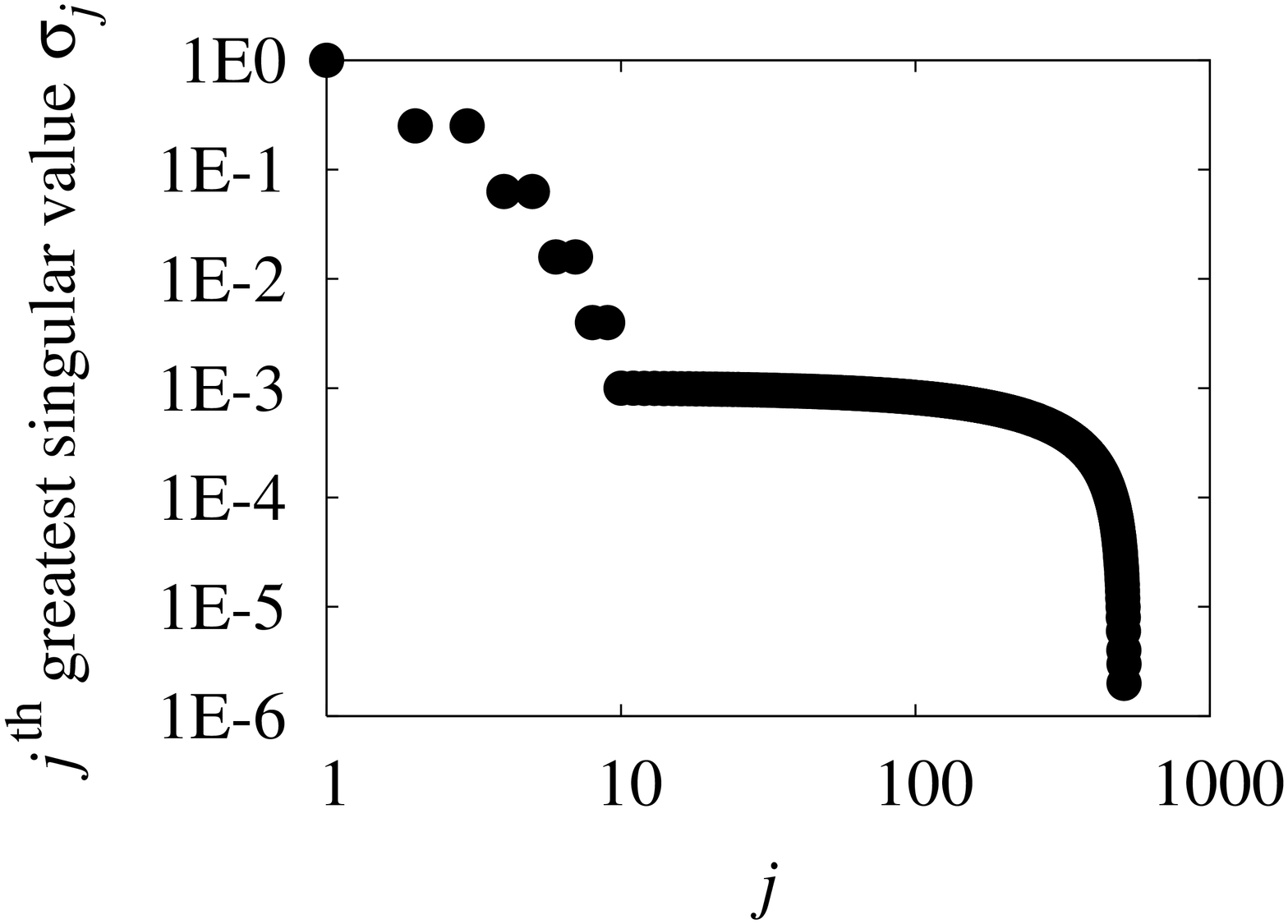}}}
\\\vspace{.15in}
Figure~1: Singular values with $m = 512$, $n = 1024$, \\
          and $\sigma_{k+1} = .001$
\end{center}
\end{figure}

\clearpage

\bibliographystyle{siam}
\bibliography{pca}

\end{document}